\documentclass{article}
\usepackage{ijcai20}
\usepackage{times}
\usepackage{framed}
\usepackage{amsfonts}
\usepackage{amsthm}
\usepackage{mathrsfs}
\usepackage{bm}

\usepackage{soul}
\usepackage{url}
\usepackage[hidelinks]{hyperref}
\usepackage[utf8]{inputenc}
\usepackage[small]{caption}
\usepackage{graphicx}
\usepackage{amsmath}
\usepackage{booktabs}
\title{Incentive-Compatible Diffusion Auctions}

\author{
Bin Li$^1$\and
Dong Hao$^1$\footnote{Corresponding author.}\and
Dengji Zhao$^2$\\
\affiliations
$^1$University of Electronic Science and Technology of China\\
$^2$ShanghaiTech University\\
\emails
\{libin@std.uestc,
haodong@uestc,
zhaodj@shanghaitech\}.edu.cn
}

\begin{document}

\maketitle

\begin{abstract}
Diffusion auction is a new model in auction design. It can incentivize the buyers who have already joined in the auction to further diffuse the sale information to others via social relations, whereby both the seller's revenue and the social welfare can be improved. Diffusion auctions are essentially nontypical multidimensional mechanism design problems and agents' social relations are complicatedly involved with their bids. In such auctions, incentive-compatibility (IC) means it is best for every agent to honestly report her valuation and fully diffuse the sale information to all her neighbors. Existing work identified some specific mechanisms for diffusion auctions, while a general theory characterizing all incentive-compatible diffusion auctions is still missing. In this work, we identify a sufficient and necessary condition for all dominant-strategy incentive-compatible (DSIC) diffusion auctions. We formulate the monotonic allocation policies in such multidimensional problems and show that any monotonic allocation policy can be implemented in a DSIC diffusion auction mechanism. Moreover, given any monotonic allocation policy, we obtain the optimal payment policy to maximize the seller's revenue.
\end{abstract}
\section{Introduction}
\newtheorem{theorem}{Theorem}
\newtheorem{defn}{Definition}
\newtheorem{lemma}{Lemma}
\newtheorem{corollary}{Corollary}
\newtheorem{claim}{Claim}

\newtheorem{prop}{Proposition}

In traditional auctions, the bidders are given no incentives to spread the sale information to others since they are in perfect competition. This not only ends the sale with a low revenue, and even worse, may lead to an inefficient allocation. For instance, suppose a seller is selling an item on eBay. Visitors who browse the sale may join in the sale and put their bids. The strike price of the item depends on the number of bidders and their posted bids. In order to involve more participants in an auction, the seller can advertise the sale via sponsored search or social media, etc. However, if the advertisements do not bring any buyers valuable enough, the investment on the advertisements could probably be wasted.

To tackle this problem, a new thread in auction design which is called \emph{diffusion auction} comes to the fore. This kind of auction mechanisms can incentivize the early joined participants to share the sale information to the uninformed individuals via, for example, online social networks. \cite{li2017mechanism} initiated this problem in the scope of social networks. They showed that the classic VCG mechanism \cite{vickrey1961counterspeculation,clarke1971multipart,groves1973incentives} can be directly extended to the social network setting, but the seller may suffer a loss by implementing this mechanism. To overcome the low revenue problem, they proposed the information diffusion mechanism (IDM) which is not efficient but guarantees the seller's revenue.
In some particular economic networks, efficiency and revenue could both be guaranteed by a natural market model called customer sharing mechanism \cite{Li2018CustomerSI}. \cite{li2019graph} further identified a class of diffusion mechanisms for any unweighted graphs, and IDM is a special case in this class with the lowest revenue. They also established a mechanism for weighted graphs.


All of the above-mentioned auctions are specific cases of incentive-compatible (IC) diffusion auctions. Nevertheless, general characterizations of incentive-compatible diffusion auctions are still missing in the literature. The most obvious advantage of incentive-compatible mechanisms is that they are easy for the participants to play and easy for the designer to predict the outcome. Moreover, according to the revelation principle \cite{myerson1979incentive}, although honest reporting may not be the only Nash equilibrium, it is without loss of generality to focus on incentive-compatible mechanisms. For diffusion auction design, since each buyer can only share the auction information to some subset of her neighbors, the argument form that is used to prove the revelation principle also holds. In this work, we characterize a \emph{sufficient and necessary condition} for incentive-compatible diffusion auctions. We prove that natural \emph{monotonic allocation} policies can be implemented in dominant-strategy incentive-compatibility (DSIC).
Moreover, given any monotonic allocation policy, we also obtain the \emph{optimal payment} policy, i.e., the one that maximizes the seller's revenue.

The characterization of incentive-compatibility can be traced back to value monotonicity \cite{Myerson1981OptimalAD} and is later rediscovered by \cite{archer2001truthful}. Value monotonicity has been widely used to obtain incentive-compatible mechanisms for single-parameter domains where each agent's preference over outcomes can be captured by a single parameter \cite{archer2001truthful,briest2005approximation,auletta2004deterministic,andelman2005truthful,kovacs2005fast}. A generalization of value monotonicity into arbitrary domains is cycle monotonicity \cite{rochet1987necessary}. However, cycle monotonicity is overly complicated and is rarely applied to argue incentive-compatibility. It is known that for convex domains, cycle monotonicity can be replaced by a simpler condition, called weak monotonicity \cite{lavi2003towards,bikhchandani2006weak,shmoys1993approximation}. But even this simpler condition does not have much application in mechanism design for multidimensional problems. When the domain of preferences is unrestricted and there are at least three possible outcomes, \cite{roberts1979characterization} shows that the only incentive-compatible mechanisms are variations of the VCG mechanism.
In the diffusion auction scenario, besides posting in a bid, each bidder is allowed to strategically diffuse the sale information to others. Therefore, this is a nontypical multidimensional problem where one domain is the bid and the other domain could be any subset (even an empty set) of edges of a graph. As we will show in the following parts, value monotonicity is not suffice to incorporate diffusion incentives and we develop new techniques for analyzing the incentive-compatibility of diffusion auctions.

Section 2 formulates the basic model of diffusion auction. Section 3 provides a general characterization of incentive-compatible diffusion auctions. In Section 4, we elaborate on a class of natural allocation policies and derive the optimal payment policies for such allocations.

\section{Preliminaries}
Consider a digraph $G=(N,E)$, where $N$ is the node set and $E$ is the edge set, assume there is a \emph{seller} $s\in N$ who wants to sell an indivisible item. Each node $i\in N\setminus\{s\}$ refers to a potential buyer and each pair of nodes can communicate through an edge in $E$. The most typical examples of such graphs are social networks where each person can only directly communicate with her neighbors, e.g., friends, relatives. Initially, only the seller's neighbors are informed of the sale. A node can participate in the sale if and only if one of her neighbors is aware of the sale and further invites her to join the sale. The seller wants to incentivize the early participants to share the sale information to the uninformed nodes, whereby all nodes in the graph can join the sale. This can not only increase the seller's revenue, but also improve the allocation efficiency. For convenience, we define $N_{-s}$ as $N\setminus\{s\}$ and call a node in $N_{-s}$ a \emph{buyer}.

The above problem can be modeled as an auction. Formally, denote buyer $i$'s privately known \emph{type} by $t_i=(v_i,r_i)$, where $v_i$ is her valuation of the item and $r_i$ is the set of all her direct neighbors. Let ${\bf t}=(t_1,\cdots,t_n)$ denote the type profile of all buyers and ${\bf t}_{-i}=(t_1,t_2,\cdots,t_{i-1},t_{i+1},\cdots,t_n)$ denote the type profile of all buyers except $i$, i.e., ${\bf t}=(t_i,{\bf t}_{-i})$. Let $T=\times T_{i\in N_{-s}}$ be the type profile space of all buyers, where $T_i=\mathcal{R}_{\geq 0}\times \mathcal{P}(N)$ is buyer $i$'s type space and $\mathcal{P}(N)$ is the power set of $N$. Since $t_i$ is private information, $i$ can strategically misreport to affect the outcome. 
Let $t_i'=(v'_i,r'_i)\in T_i$ be the reported type of buyer $i$ where $r'_i\subseteq r_i$ means that $i$ diffuses the sale information to a subset $r'_i$ of her neighbors. Since buyers can not diffuse information to a non-existing neighbor, we have $r'_i \in \mathcal{P}({r_i})$. Denote a \emph{feasible} reported type profile of all buyers by ${\bf t'}=(t_1',\cdots,t_n')$, where feasibility requires that $t'_i=nil$ if $i$ has never been informed of the sale or she decides not to join the sale. We only consider feasible reported type profiles.


\begin{defn}
	A diffusion auction mechanism $\mathcal{M}=(\pi,x)$ on graph $G$ consists of two components: an {allocation policy} $\pi = \{\pi_i\}_{i\in N_{-s}}$ and a {payment policy} $x=\{x_i\}_{i\in N_{-s}}$, where $\pi_i:T\rightarrow \{0,1\}$ and $x_i:T\rightarrow \mathcal{R}$ are the allocation and payment functions for $i$, respectively.
\end{defn}

Given buyers' reported type profile ${\bf t}^\prime$, $\pi_i({\bf t}^\prime) = 1$ means that $i$ gets the item, while $i$ does not get the item if $\pi_i({\bf t}^\prime) = 0$. $x_i({\bf t}^\prime) \geq 0$ indicates that $i$ pays the seller $x_i({\bf t}^\prime)$, and $i$ receives $|x_i({\bf t}^\prime)|$ from the seller if $x_i({\bf t}^\prime) < 0$. We say an allocation policy $\pi$ is \emph{feasible} if for all ${\bf t}' \in T$, $\sum_{i\in N_{-s}} \pi_i({\bf t}') \leq 1$, and if $\pi_i({\bf t}') = 1$, then $t'_i \neq nil$. That is, a feasible allocation can only allocate the item to at most one participant. Let $\Pi$ be the set of all feasible allocations. In what follows, we only consider feasible allocations. Given an allocation policy $\pi$ and a type profile ${\bf t}$, the \emph{social welfare} is define as $W(\pi,{\bf t})=\sum_{i\in N_{-s}} \pi_i({\bf t})v_i$. An allocation policy ${\pi^*}$ is efficient if for all ${\bf t}\in T$ it always maximizes $W({\pi^*},{\bf t})$. 
\begin{defn}
	An allocation ${\pi^*}$ is {efficient} if for all ${\bf t} \in T$, ${\pi^*} \in {\arg\max}_{{\pi^\prime} \in \Pi} W(\pi', {\bf t})$.
\end{defn}

Given buyer $i$ with truthful type $t_i = (v_i, r_i)$, a reported type profile ${\bf t}^\prime$ and a mechanism $(\pi, x)$, the quasilinear \emph{utility} of $i$ is defined as: $u_i(t_i, {\bf t}^\prime, (\pi, x)) = \pi_i({\bf t}^\prime)v_i  - x_i({\bf t}^\prime).$ We say a diffusion auction is individually rational if for each buyer, her utility is non-negative when she truthfully reports her valuation, no matter to which subset of neighours she diffuses the auction information and what the others do.
\begin{defn}
	A diffusion auction mechanism $(\pi, x)$ is {individually rational} (IR) if $u_i(t_i, ((v_i, r_i^\prime),{\bf t}_{-i}^\prime), (\pi, x)) \geq 0$ for all $i\in N_{-s}$, all $r_i^\prime \in \mathcal{P}(r_i)$, and all ${\bf t}_{-i}^\prime \in T_{-i}$.
\end{defn}
In diffusion auctions, IR ensures that no matter how a buyer diffuses the information, as long as she honestly reports her value, she will not suffer loss. Besides IR, now we define the incentive-compatibility (IC) as: for each buyer, reporting her true valuation and, meanwhile, fully diffusing the sale information to all her neighbors is always a dominant strategy, no matter what the others do.
\begin{defn}
	A diffusion auction mechanism $(\pi, x)$ is {incentive-compatible} (IC) if
	$u_i(t_i, (t_i,{\bf t}_{-i}^\prime), (\pi, x)) \geq u_i(t_i, (t_i^\prime, {\bf t}_{-i}^{\prime\prime}), (\pi, x))$ for all $i\in N_{-s}$, all $t_i^\prime \in T_i$, and all ${\bf t}_{-i}^\prime \in T_{-i}$.
\end{defn}
Note that on the right-hand side of the inequality, $(t_i, {\bf t}_{-i}^{\prime})$ is replaced by $(t_i^\prime, {\bf t}_{-i}^{\prime\prime})$. This is because if $i$ does not spread the auction information to all her neighbors, then some buyers who can receive the information under $t_i$, may no longer be able to receive it. Thus under $t_i^\prime$, the feasible type profile of buyers except $i$ is changed from ${\bf t}_{-i}^{\prime}$ to ${\bf t}_{-i}^{\prime\prime}$.

Now we formulate the diffusion auctions' criteria with respect to the seller. Given a reported type profile ${\bf t}^\prime$ and a diffusion auction mechanism $\mathcal{M} = (\pi, x)$, the seller's \emph{revenue} generated by $\mathcal{M}$ is defined as the sum of all buyers' payments $Rev^{\mathcal{M}}({\bf t}^\prime) = \sum_{i\in N_{-s}} x_i({\bf t}^\prime)$.
\begin{defn}
	A diffusion auction mechanism $\mathcal{M}=(\pi, x)$ is {(weakly) budget balanced} if for all ${\bf t}' \in T$, $Rev^{\mathcal{M}}({\bf t}') \geq 0$.
\end{defn}

Over the past few years, following similar models as the above, several diffusion auction mechanisms have been proposed, while a general theory for diffusion auctions is missing. In the following sections, we will establish such a theory to fill the gap.

\section{Incentive-Compatible Diffusion Auctions}
We say an allocation policy $\pi$ is \emph{implementable} if there is a payment policy $x$ such that the mechanism $(\pi,x)$ is IR and IC, and say it is value-monotonic if a winner cannot become a loser by increasing her bid. Given a value-monotonic allocation policy, there exists a critical bid for buyer $i$, above/below which $i$ wins/loses the item, i.e., each buyer's critical bid is the minimum bid that makes her win.

The most simple scenario is that seller $s$ connects to all buyers in $N_{-s}$. Since the seller knows all buyers in advance, it reduces to the traditional single-item auction design. General theories for this case have been well developed. Many well-known works \cite{Myerson1981OptimalAD,lehmann2002truth,archer2001truthful,mu2008truthful} show that a normalized mechanism (i.e., losers always pay zero) for single-parameter domains is IC if and only if its allocation policy is value-monotonic, and the winner pays exactly her critical bid. The most well-know theory for this case is Myerson's Lemma \cite{Myerson1981OptimalAD}.



\begin{theorem}[Myerson's Lemma]\label{myerson}
	Fix a single-item auction environment without information diffusion.
	\begin{itemize}
		\item[(a)] 	An allocation policy $\pi$ is implementable if and only if it is value-monotonic.
		\item[(b)] If $\pi$ is value-monotonic, then there is a unique payment policy $x$ for which $(\pi,x)$ is IR and IC and the winner pays her critical bid and the losers pay zero.
	\end{itemize}
\end{theorem}
However, Myerson's Lemma does not apply in our general setting.
First of all, although critical bid is a powerful tool of characterizing the payment policy in traditional auctions, the characterization of it is not straightforward in diffusion auctions. In single-parameter domains, critical bid is well defined which depends only on other bidders' bids. For example, in the second price auction, each bidder's critical bid equals the highest bid of the other buyers which is not affected by herself. Unfortunately, this argument is no longer valid in diffusion auctions, since any buyer $i$ now can affect others' report ${\bf t}'_{-i}$ through strategic diffusion. Secondly, the losers' payoffs in a diffusion auction can be positive and can vary from one bidder to another \cite{li2019graph}.

In diffusion auctions, any buyer $i$'s diffusion strategy can affect the allocations. Given any $r_i'\subseteq r_i$ there is a minimum winning bid for buyer $i$. That is, each buyer $i$'s critical bid should be a function of her diffusion strategy $r_i'$. Define the critical bid in diffusion auctions as follows.
\begin{defn}
	Given an allocation policy $\pi$ and other buyers' reports ${\bf t}'_{-i}$, buyer $i$'s critical bid with respect to $r_i'\subseteq r_i$ is $v^*_i(r_i')=\arg \min_{b_i'\in \mathcal{R}_{\geq 0}}\{\pi_i((b_i',r'_i),{\bf t}'_{-i})=1\}$.
\end{defn}

Note that under some special allocation policies, a buyer cannot win no matter what she does. An immediate evidence is the 'unlucky buyers' identified in \cite{li2017mechanism}. It is a special class of buyers who cannot affect the allocations no matter how high they bid and which subset of neighbors they diffuse the sale information to. Whenever in these cases, set $v^*_i(r_i')=\infty$ for all $r_i'\subseteq r_i$. Whether such buyers exist is jointly determined by the structure of the graph and the designed allocation policy.


For single-item auctions, a buyer either wins the item or not. Hence, a buyer's payment policy $x_i({\bf t'})$ can be decoupled into two components: the payment for winning the item $x_i({\bf t'}|\pi_i({\bf t'})=1)$ and the payment for losing it $x_i({\bf t'}|\pi_i({\bf t'})=0)$. Let $\tilde{x}_i(t_i',{\bf t}'_{-i})=x_i({\bf t'}|\pi_i({\bf t'})=1)$ and $\overline{x}_i(t_i',{\bf t}'_{-i})=x_i({\bf t'}|\pi_i({\bf t'})=0)$, respectively. The following proposition is straightforward from this definition and will be very useful to explore the underlying structures of diffusion auction design.
\begin{prop}[Payment Decoupling]\label{Payment-Decoupling}
Given a mechanism $(\pi,x)$, for all $i$ and all ${\bf t}'_{-i}$, the payment for $i$ can be decoupled as $x_i({\bf t'})=\pi_i({\bf t'})\tilde{x}_i(t_i',{\bf t}'_{-i})+(1-\pi_i({\bf t'}))\overline{x}_i(t_i',{\bf t}'_{-i})$. 
\end{prop}

We call $\tilde{x}$ and $\overline{x}$ the \emph{decoupled-payments} from $x$. To simplify the notations, the term ${\bf t}_{-i},{\bf t}_{-i}'$ will be omitted when analyzing a specific buyer $i$. In the following definitions and lemmas, we will characterize several properties that any IC diffusion auction should possess, and then in Theorem \ref{maint} we will further prove that these properties are also sufficient for any diffusion auction to be IC. We first extend the traditional value-monotonicity into diffusion auctions.
\begin{defn}[Value-Monotonic Allocation]\label{def-value-mono}
	An allocation policy $\pi$ for a diffusion auction is value-monotonic if for any ${\bf t}'$ and any buyer $i$ with $\pi_i((v_i',r_i'),{\bf t}_{-i}')=1$, we have $\pi_i((v_i'',r_i'),{\bf t}_{-i}')=1$ when $v_i''\geq v_i'$.
\end{defn}
It means that for any buyer $i$, once her diffusion action is fixed as $r_i'$ and all the other buyers' reports are fixed as ${\bf t}_{-i}'$, then her allocation is monotonically increasing over her bid $v_i'$. In other words, given ${\bf t}_{-i}'$ and $r_i'$, $i$ wins the item if and only if $v_i'\geq v_i^*(r_i')$.
We now prove that any IC diffusion auction should have value-monotonic allocation.
\begin{lemma}\label{p1}
	If a diffusion auction $(\pi,x)$ is incentive-compatible, then $\pi$ is value-monotonic.
\end{lemma}
\begin{proof}
Proof by contradiction: assume there exists an IC diffusion auction and $\pi$ is not value-monotonic. This means there are two values $v_i^2>v_i^1$ such that $\pi_i((v_i^1,r_i),{\bf t}_{-i})=1$ but $\pi_i((v_i^2,r_i),{\bf t}_{-i})=0$. Denote the payments for buyer $i$ with types $t_i^1=(v_i^1,r_i)$ and $t_i^2=(v_i^2,r_i)$ by $\tilde{x}_i(t_i^1)$ and $\overline{x}_i(t_i^2)$, respectively. Recall that $\tilde{x}_i$ is $i$'s payment for winning and $\overline{x}_i$ is her payment for losing. If $v_i^1-\tilde{x}_i(t_i^1)\geq -\overline{x}_i(t_i^2)$, then buyer $i$ with losing type $t_i^2$ should misreport $t_i^1$ since $v_i^2-\tilde{x}_i(t_i^1)> v_i^1-\tilde{x}_i(t_i^1)\geq -\overline{x}_i(t_i^2)$; similarly, if $v_i^1-\tilde{x}_i(t_i^1)< -\overline{x}_i(t_i^2)$, then buyer $i$ with winning type $t_i^1$ should misreport $t_i^2$ to lose. Since there always exists a beneficial deviation, then the diffusion auction cannot be IC and the assumption is false.
\end{proof}

We next extend the traditional value-independent payment into diffusion auctions based on the payment-decoupling in Proposition \ref{Payment-Decoupling}.

\begin{defn}[Bid-Independent Decoupled-Payments]
The two decoupled-payments from $x$ are bid-independent if for all ${\bf t}'_{-i}$, all $i$ and $i$'s two different reports $(v_i',r_i')$ and $(v_i'',r_i')$, we have $\tilde{x}_i(v_i',r_i')=\tilde{x}_i(v_i'',r_i')$ and $\overline{x}_i(v_i',r_i')=\overline{x}_i(v_i'',r_i')$.
\end{defn}
The payment for $i$ consists of the two decoupled-payments, thus if both of them are bid-independent, then the payment for $i$ is also bid-independent. This covers the property of payment in traditional single-item auctions. We now prove that any incentive-compatible diffusion auction should have bid-independent decoupled-payments.


\begin{lemma}\label{p2}
	If a diffusion auction $(\pi,x)$ is incentive-compatible, then the decoupled-payments of $x$ are bid-independent.
\end{lemma}
\begin{proof}
Given a buyer $i$ and ${\bf t}_{-i}$, for any two winning types $t_i^1=(v_i^1,r_i)$ and $t_i^2=(v_i^2,r_i)$ where $\min\{v_i^1,v_i^2\}\geq v^*_i(r_i)$, assume for a contradiction that $\tilde{x}_i(v_i^1,r_i)\neq \tilde{x}_i(v_i^2,r_i)$, then buyer $i$ with $t_i^1$ would misreport $t_i^2$ when $\tilde{x}_i(v_i^1,r_i)> \tilde{x}_i(v_i^2,r_i)$, and buyer $i$ with $t_i^2$ should misreport $t_i^1$ when $\tilde{x}_i(v_i^1,r_i)<\tilde{x}_i(v_i^2,r_i)$. Therefore, for the winning buyer, her payment $\tilde{x}_i$ should be independent of her bid. The same argument is also true for the losers. If $\max\{v_i^1,v_i^2\}< v^*_i(r_i)$ and $\overline{x}_i(v_i^1,r_i)\neq \overline{x}_i(v_i^2,r_i)$, then loser $i$ with higher payment may misreport. Hence, $\tilde{x}_i$ and $\overline{x}_i$ should be independent of buyer's bid in any IC diffusion auction.
\end{proof}

Since $\tilde{x}_i$ and $\overline{x}_i$ are independent of $v_i'$ in any IC diffusion auction, we replace $\tilde{x}_i(v_i',r_i)$ and $\overline{x}_i(v_i',r_i)$ with $\tilde{x}_i(r_i)$ and $\overline{x}_i(r_i)$ for expression convenience.
The following lemma further deduces that if a diffusion auction is IC, then for every buyer, her critical payment should be a binding constraint for her decoupled-payments.
\begin{lemma}\label{p3}
	If a diffusion auction $(\pi,x)$ is incentive-compatible, then for any ${\bf t}$ and all $i$, the equation  $\tilde{x}_i(r_i)-\overline{x}_i(r_i)=v^*_i(r_i)$ holds.
\end{lemma}
\begin{proof}
	Assume there is an IC diffusion auction and $\tilde{x}_i(r_i)-\overline{x}_i(r_i)\neq v^*_i(r_i)$. Then either $\tilde{x}_i(r_i)-\overline{x}_i(r_i)> v^*_i(r_i)$ or $\tilde{x}_i(r_i)-\overline{x}_i(r_i)< v^*_i(r_i)$. For the former case, the winning buyer with $t_i=(v^*_i(r_i),r_i)$ will try to lose by lowering her bid, since her utility will be increased from $v^*_i(r_i)-\tilde{x}_i(r_i)$ to $-\overline{x}_i(r_i)$. Similarly, for the latter case, a losing buyer $i$ with $t_i=(v_i,r_i)$, where $\tilde{x}_i(r_i)-\overline{x}_i(r_i)<v_i<v^*_i(r_i)$, will misreport to win since her utility of winning $v_i-\tilde{x}_i(r_i)$ is larger than that of losing $-\overline{x}_i(r_i)$. In either case, it contradicts the assumption.
\end{proof}
The last property we want to present is that in an IC diffusion auction, a buyer's payment should be minimized by diffusing the sale information to all her neighbors. 
\begin{defn}[Diffusion-Monotonic Decoupled-Payments]
The two decoupled-payments from $x$ are diffusion-monotonic if for all $i$ and all ${\bf t}'_{-i}$, when $r_i''\subseteq r_i'$, we have $\tilde{x}_i(v_i', r_i'')\geq \tilde{x}_i(v_i', r_i')$ and $\overline{x}_i(v_i', r_i'')\geq \overline{x}_i(v_i', r_i')$.
\end{defn}
The following lemma guarantees that all IC diffusion auctions possess this diffusion-monotonicity property.
\begin{lemma}\label{p4}
	If a diffusion auction $(\pi,x)$ is incentive-compatible, then the decoupled-payments from $x$ are diffusion-monotonic.
\end{lemma}
\begin{proof}
Consider an IC diffusion auction and two diffusion sets $r_i^2\subseteq r_i^1$. If $\tilde{x}_i(r_i^2)<\tilde{x}_i(r_i^1)$, then a winning buyer with $(v_i,r_i^1)$, where $v_i\geq \max\{v^*_i(r_i^1),v^*_i(r_i^2)\}$, would misreport $(v_i,r_i^2)$ since $v_i-v^*_i(r_i^2) > v_i-v^*_i(r_i^1)$. This contradicts IC, thus $\tilde{x}_i(r_i^2)<\tilde{x}_i(r_i^1)$ is false. Similar argument shows $\overline{x}_i(r_i^2)<\overline{x}_i(r_i^1)$ is also false.
\end{proof}
Next, we prove that the above properties together also serve as a sufficient condition for IC diffusion auctions.

\begin{theorem}\label{maint}
	A diffusion auction $(\pi,x)$ is incentive-compatible if and only if for all type profile ${\bf t}$ and all $i$, P1-P4 are satisfied, where
	\begin{itemize}\itemindent=2em
		\item [P1]: $\pi$ is value-monotonic,
		\item [P2]: $\tilde{x}_i$ and $\overline{x}_i$ are bid-independent,
		\item [P3]: $\tilde{x}_i(r_i)-\overline{x}_i(r_i)=v^*_i(r_i)$,
		\item [P4]: $\tilde{x}_i$ and $\overline{x}_i$ are diffusion-monotonic.
	\end{itemize}
\end{theorem}
\begin{proof}
	Lemma \ref{p1}-\ref{p4} proved the necessity of these four properties in any IC diffusion auction. Next we prove that if a mechanism satisfies P1-P4, then it is IC. The proof includes two steps. The first step shows that the winner has no incentives to deliberately become a loser and vise versa. The second step shows once a buyer's allocation is determined, truthfully reporting maximizes her utility .
	
According to P2, the bids in all $\tilde{x}_i$ and $\overline{x}_i$ are omitted. If $i$ is the winner by truthfully reporting, when she misreports to lose, her utility changes from $v_i-\tilde{x}_i(r_i)$ to $-\overline{x}_i(r_i')$, where $r_i'\subseteq r_i$. We have $v_i-\tilde{x}_i(r_i)\geq v^*_i(r_i)-\tilde{x}_i(r_i)=-\overline{x}_i(r_i)\geq -\overline{x}_i(r_i')$, where the left inequality comes from that $i$ wins by truthfully reporting, the inner equation comes from P3 and the right inequality comes from P4. Therefore, a winner has no incentive to misreport to become a loser. If $i$ is a loser by truthfully reporting, when she misreports to win, her utility changes from $-\overline{x}_i(r_i)$ to $v_i-\tilde{x}_i(r_i')$. We have $v_i-\tilde{x}_i(r_i')\leq v^*_i(r_i)-\tilde{x}_i(r_i')\leq v^*_i(r_i)-\tilde{x}_i(r_i)= -\overline{x}_i(r_i)$ where the first inequality comes from P1, the second inequality comes from P4 and the equation is from P3. Therefore a loser has no incentive to misreport to win.

Now we prove that each kind of buyers maximize their utilities by truthfully reporting. A winner's utility is $v_i-\tilde{x}_i(r_i)$. If she misreports $t_i'$ and still wins, her utility becomes $v_i-\tilde{x}_i(r_i')$ which is no more than $v_i-\tilde{x}_i(r_i)$ according to P2 and P4. Also, reporting $t_i'$ has a risk of becoming a loser, which is not beneficial. Therefore it is always an optimal strategy for the winner to be truthful. Similar arguments also hold for the losers. In a word, either a buyer wins or loses, being truthful is always a dominant strategy. That is, P1-P4 is a sufficient condition for a mechanism to be IC.
\end{proof}
Besides IC, as we mentioned earlier in Definition 3, a diffusion auction should also be individually rational (IR) which guarantees that each buyer's utility is nonnegative when bidding truthfully.
\begin{theorem}\label{ir}
	An incentive-compatible diffusion auction is individually rational if and only if for all buyer $i$ and all type profile ${\bf t}$, property P5: $\overline{x}_i(\emptyset)\leq 0$ is satisfied.
\end{theorem}
\begin{proof}
	Fix ${\bf t}_{-i}$, then $i$'s utility is $u_i(v_i,r_i)=\pi_i(v_i,r_i)(v_i-\tilde{x}_i(r_i))+(1-\pi_i(v_i,r_i))(-\overline{x}_i(r_i))=\pi_i(v_i,r_i)(v_i-v^*(r_i))-\overline{x}_i(r_i)$, where the last equation comes from P3. IR means that for any true type $(v_i',r_i')\in \mathcal{R}_{\geq 0}\times \mathcal{P}({r_i})$, the inequality $\pi_i(v_i',r_i')(v_i'-v^*(r_i'))-\overline{x}_i(r_i')\geq 0$ must hold. This can be guaranteed when $\min_{v'_i\in \mathcal{R}_{\geq 0},r_i'\subseteq r_i}\pi_i(v'_i,r'_i)(v'_i-v^*(r'_i))-\overline{x}_i(r_i')\geq 0$ holds. By P1, $\pi_i(v'_i,r'_i)=1$ when $v'_i\geq v^*(r'_i)$ and $\pi_i(v'_i,r'_i)=0$ when $v'_i< v^*(r'_i)$, therefore the minimization requires the bid $v'_i< v^*(r'_i)$. Then the minimization function becomes $-\overline{x}_i(r_i')$. According to P4, $-\overline{x}_i(r_i')$ is minimized by choosing $r_i'=\emptyset$. Hence, the minimum of $u_i$ is $-\overline{x}_i(\emptyset)$, where $t_i=(v_i'< v^*_i(\emptyset),r'_i=\emptyset)$. Therefore, IR is equivalent to the property that for each $i$ and each ${\bf t}$, $\overline{x}_i(\emptyset)\leq 0$ must hold.
\end{proof}

We emphasize that Theorem \ref{mal} and \ref{ir} imply Myerson’s Lemma, but not vice versa. Given any classic value-monotonic allocation policy, it is not guaranteed that a mechanism will have diffusion incentives. This is why we must have properties P3-P5 hold.
The above results present the IC and IR diffusion auctions in a high-level view, while no details are given for the allocation and payment policies.
Since a diffusion auction can be very complex, even when the allocation is value-monotonic, it could be very difficult to obtain the corresponding payment policy.
In the next section, we present a class of more operational and natural allocation policies, for which an elegant form of the payment policies can be achieved.


\section{Monotonic Allocation and Optimal Payment}
Our design of natural allocation policies is inspired by the following intuition:
in the real world, if a person wins in a sale, she can still win by posting a higher bid; a person who diffuses the sale information to less friends has a better chance of winning the sale, since some high-bid competitors might be excluded from the sale.
We focus on such allocations.

Formally, given a buyer $i$, define a partial order $\succeq_{t_i}$ on types: $(v_i^1,r_i^1)\succeq (v_i^2,r_i^2)$ if and only if $v_i^1\geq v_i^2 \wedge r_i^1\subseteq r_i^2$.
\begin{defn}
	An allocation policy $\pi$ for a diffusion auction is monotonic if for every type $t_i$ with $\pi_i(t_i,{\bf t}_{-i})=1$, we have that $\pi_i(t_i',{\bf t}'_{-i})=1$ for any $t_i'\succeq t_i$.
\end{defn}
This definition is an extension of value-monotonicity in Definition \ref{def-value-mono}.
The following lemma presents a nice property of monotonic allocation policies, it is a key component of the results in this section.

\begin{lemma}\label{cbm}
	If an allocation policy $\pi$ is monotonic, then for all $i$ and all $r'_i\subseteq r_i$, $v_i^*(r_i')\leq v_i^*(r_i)$.
\end{lemma}
\begin{proof}
	Assume there is a monotonic $\pi$ such that for two reports $t'_i=(v^*_i(r_i), r'_i)$ and $t_i=(v^*_i(r_i),r_i)$, where $r_i'\subseteq r_i$, it has $v^*_i(r'_i)>v^*_i(r_i)$. In this case, we have $t'_i\succ t_i$, however, $\pi_i(t'_i)=0$ and $\pi_i(t_i)=1$, which contradicts the monotonicity. Thus the assumption is false.
\end{proof}
%
The key observation from Lemma \ref{cbm} is that under any monotonic allocation policy, buyer $i$'s critical bid function $v_i^*(r_i)$ is nondecreasing in $r_i$. Recall that in characterizations of incentive-compatibility, P4 shows both $\tilde{x}_i$ and $\overline{x}_i$ are nonincreasing in $r_i$. Moreover, P3 shows $\tilde{x}_i(r_i)-\overline{x}_i(r_i)=v_i^*(r_i)$. Based on Lemma \ref{cbm} and the above observations, we show that any monotonic allocation policy is implementable.

\begin{prop}\label{mal}
	In diffusion auctions, any monotonic allocation policy is implementable.
\end{prop}
\begin{proof}
	Remind that a diffusion auction is IR and IC if and only if for any type profile ${\bf t}_{-i}$ and any buyer $i$ with $(v_i,r_i)$, P1-P5 hold.
	Given any monotonic allocation policy, P1 is directly met according to the definitions. Construct two decoupled-payments as $\tilde{x}_i(r_i)=\tilde{f}_i(v_i^*(r_i))+h_i({\bf t}_{-i})$ and $\overline{x}_i(r_i)=\overline{f}_i(v_i^*(r_i))+h_i({\bf t}_{-i})$, where $h_i({\bf t}_{-i})$ is independent of $i$'s report, and $\tilde{f}_i(v_i^*(r_i))$ and $\overline{f}_i(v_i^*(r_i))$ are nonincreasing in $v_i^*(r_i)$, and $\tilde{f}_i(v_i^*(r_i))-\overline{f}_i(v_i^*(r_i))=v_i^*(r_i)$. This construction for $\tilde{x}$ and $\overline{x}$ satisfies P2-P4. A possible choice is $\tilde{f}_i(v_i^*(r_i))=-\alpha v_i^*(r_i)$ and $\overline{f}_i(v_i^*(r_i))=-(1+\alpha)v_i^*(r_i)$, where $\alpha$ could be any nonnegative value. Setting $h_i({\bf t}_{-i})=(1+\alpha)v_i^*(\emptyset)$, P5 is also satisfied. Therefore, for any monotonic allocation, there is a payment such that the mechanism is IR and IC. This completes the proof.
\end{proof}

From the proof of Proposition \ref{mal}, we can see that even when given a monotonic allocation policy, the space of proper payment policies could still be very large. For example, $\tilde{f}_i$ and $\overline{f}_i$ need not be linear in $v_i^*(r_i)$, they can be any polynomials with negative coefficients.
In the following we select payment policies with respect to the seller's revenue.
\begin{defn}
	Given a monotonic allocation policy $\pi$ and two payment policies ${x^1},{x^2}$ satisfying IR and IC, we say ${x^1}$ dominates ${x^2}$ in the revenue if for any type profile ${\bf t}$, $Rev^{(\pi,{x^1})}({\bf t})\geq Rev^{(\pi,{x^2})}({\bf t})$.
\end{defn}
If there exists a payment policy ${x^*}$ such that for any type profile ${\bf t}\in T$, it dominates {\it all} the other payment policies under IR and IC conditions, then ${x^*}$ is optimal with respect to $\pi$. Our next result shows that for any monotonic allocation policy, the optimal payment policy exists.

\begin{theorem}\label{mam}
	Given any monotonic allocation policy $\pi$, the payment policy ${x^*}=\{x^*_i=v_i^*(\emptyset)-v^*_i(r_i)(1-\pi_i({\bf t}))\}_{i\in N_{-s}}$ is optimal.
\end{theorem}
\begin{proof}
	Given any type report profile ${\bf t}$ and a monotonic allocation policy $\pi$, the seller's revenue can be expressed as
	$$Rev({\bf t})=\tilde{x}_w(r_w)+\sum_{i(\neq w)\in N_{-s}}\overline{x}_i(r_i), $$
	where $w$ is the winner and $i\neq w$ are the losers. The optimal payment policy $x^*$ should maximize $Rev({\bf t})$ and $(\pi, x^*)$ should be IR and IC.
For all IR and IC mechanisms, P1-P5 should hold. P1 is met by the monotonic allocation. According to P3, $\tilde{x}_w(r_w)=\overline{x}_w(r_w)+v^*_w(r_w)$. Substituting this equation into $Rev({\bf t})$, we get $v_w^*(r_w)+\sum_{i\in N_{-s}}\overline{x}_i(r_i)$. Since the allocation policy is already given, $v_w^*(r_w)$ is a constant for any ${\bf t}\in T$. Hence, the objective now is to find $x^*=\arg {\max _x}\sum_{i\in N_{-s}}\overline{x}_i(r_i)$ where P2-P5 hold for all $i$.
Furthermore, since each buyer's payment policy $(\tilde{x}_i,\overline{x}_i)$ is independent of the others' payment policies, the problem of searching for $argmax_{x=(\tilde{x}_i,\overline{x}_i)} \sum_{i\in N_{-s}}\overline{x}_i(r_i)$ becomes solving $|N|-1$ independent maximization subproblems:
	\begin{alignat}{2}
	\max_{(\tilde{x}_i,\overline{x}_i)} \quad & \overline{x}_i(r_i) \nonumber\\
	\mbox{s.t.}\quad &\forall v_i'\neq v_i'', r_i'\subseteq r_i,\nonumber\\
	&\overline{x}_i(v_i',r_i)=\overline{x}_i(v_i'',r_i) ,& \tag{1}\\
	&\tilde{x}_i(r_i)=\overline{x}_i(r_i)+v_i^*(r_i),&\tag{2}\\
	&\overline{x}_i(r_i)+v_i^*(r_i)\leq \overline{x}_i(r_i')+v_i^*(r_i'),&\tag{3}\\
	&\overline{x}_i(\emptyset)\leq 0.&\tag{4}
	\end{alignat}
	Since $v_i^*(r_i)$ is independent of $i$'s bid, constraints (1) and (2) are equivalent to P2. Constraints (2) and (3) together with Lemma \ref{cbm} ensure that $\tilde{x}_i$ and $\overline{x}_i$ are diffusion-monotonic. Constraint (4) is from P5. Hence, constraints (1)-(4) are equivalent to P2-P5, provided that $\pi$ is monotonic. This means a solution to the above system is an implementable allocation that maximizes the seller's revenue.
	
According to constraint (3), we must have $\overline{x}_i(r_i)+v_i^*(r_i)\leq \overline{x}_i(\emptyset)+v_i^*(\emptyset)$. Rearranging the order, we have $$\overline{x}_i(r_i)\leq \overline{x}_i(\emptyset)+v_i^*(\emptyset)-v_i^*(r_i).$$ Note that the left side of this inequality is just one of our maximization subproblems. Since $\overline{x}_i(\emptyset)\leq 0$ according to constraint (4), we know that $\overline{x}_i(r_i)$ is bounded by $v_i^*(\emptyset)-v_i^*(r_i)$ which is a constant with respect to an allocation policy $\pi$ and ${\bf t}$. Our finding is that $$\overline{x}_i(r_i)=v_i^*(\emptyset)-v_i^*(r_i)$$ and $$\tilde{x}_i(r_i)=\overline{x}_i(r_i)+v_i^*(r_i)=v_i^*(\emptyset)$$ is a feasible solution and hence the optimal one. To see this, notice that $\overline{x}_i(\emptyset)$ is always $0$ when $\overline{x}_i(r_i)=v_i^*(\emptyset)-v_i^*(r_i)$. Therefore constraint (4) is satisfied. Since $v_i^*(r_i)$ is independent of $v_i'$ and is also nondecreasing in $r_i$, constraints (1) and (3) are satisfied. Finally, the maximized value of $\overline{x}_i(r_i)$ and the corresponding value of $\tilde{x}_i(r_i)$ are $\overline{x}_i(r_i)=v_i^*(\emptyset)-v_i^*(r_i)$ and
$\tilde{x}_i(r_i)=v_i^*(\emptyset)$, respectively. That is, charging the winner $v_i^*(\emptyset)$ and the losers $v_i^*(\emptyset)-v_i^*(r_i)$ will maximize the seller's revenue. Based on Proposition \ref{Payment-Decoupling}, we have $x^*_i=\tilde{x}_i(r_i)\pi_i({\bf t})+\overline{x}_i(r_i)(1-\pi_i({\bf t}))=v^*_i(\emptyset)-v^*_i(r_i)(1-\pi_i({\bf t}))$. 
Now we complete the proof.
\end{proof}
Under the above framework, it is not difficult to verify that the allocation policies of the recently proposed diffusion auction mechanisms \cite{li2017mechanism,Li2018CustomerSI,li2019graph} are all monotonic and their payment policies implicitly satisfy the form in Theorem \ref{mam}. Especially, the classic VCG mechanism uses an efficient allocation that is monotonic. \cite{li2017mechanism} proved that under the diffusion auction scenario, the VCG mechanism is not budget-balanced, meaning that the seller may have a large deficit. We emphasize that the payment policy in the VCG mechanism also satisfies the optimality in Theorem \ref{mam}. Then we know that if the allocation policy is efficient, then any IR, IC diffusion auction is not weakly budget-balanced.
\begin{prop}\label{vcg_bb}
	There exists no IR and IC diffusion auction which is efficient and (weakly) budget-balanced.
\end{prop}
Before proving this proposition, we first show that the payment policy of the VCG mechanism follows Theorem \ref{mam}.
\begin{lemma}\label{vcg_p}
	In the VCG mechanism, each buyer's payment satisfies the from in Theorem \ref{mam}.
\end{lemma}
\begin{proof}
	In the VCG mechanism, the highest bidder wins the item. Clearly, this allocation policy is monotonic. Each buyer $i$ in the VCG mechanism pays $$x_i^{vcg}({\bf t}) = W(\pi^*, {\bf t}\setminus\{t_i\}) - (W(\pi^*, {\bf t})- \pi_i^*({\bf t})v_i),$$ which is interpreted as the social welfare decrease of others caused by $i$'s participation. Since there is only one item for sale, then $W(\pi^*, {\bf t}\setminus\{t_i\})$ is the highest bid without $i$'s participation and it exactly equals $v_i^*(\emptyset)$. If $i$ does not win, then the second term $W(\pi^*, {\bf t})- \pi_i^*({\bf t})v_i =W(\pi^*, {\bf t})=v^*_i(r_i)$. Otherwise, the second term is $0$. Hence, for any buyer $i$, $$\tilde{x}^{vcg}_i=W(\pi^*, {\bf t}\setminus\{t_i\})=v_i^*(\emptyset)$$ and $$\overline{x}^{vcg}_i=W(\pi^*, {\bf t}\setminus\{t_i\})-W(\pi^*, {\bf t})=v_i^*(\emptyset)-v_i^*(r_i).$$ Such a payment policy matches the form in Theorem \ref{mam}.
\end{proof}
\begin{proof}[Proof of Proposition \ref{vcg_bb}]
	According to Lemma \ref{vcg_p} and Theorem \ref{mam}, we know that the payment policy in the VCG mechanism maximizes the seller's revenue. Since the allocation policy is efficient and \cite{li2017mechanism} shows that the VCG mechanism is not budget balanced, Proposition \ref{vcg_bb} is true.
\end{proof}
%

\section{Conclusions}
In this paper, we characterize a sufficient and necessary condition for all incentive-compatible and individually rational diffusion auctions.
We also propose a class of natural monotonic allocation policies, for which we obtain the corresponding optimal payment policy that maximizes the seller's revenue.
In practice, diffusion auctions can be applied in a recursive way: any new comer will be treated as a bidder if she submits a two-dimensional type. This recursion is terminated if there is no new submission. Based on the reports by all participants, the seller can build the social network and determines the allocation and the payments.

An immediate extension of this work is IC in diffusion auctions where the seller has multiple homogeneous items \cite{Zhao2018Multi,kawasaki2019strategy}. It is also important to investigate how the relaxation of IC constraints can improve the efficiency in diffusion auctions \cite{jeong2020groupwise}. Beyond the scope of diffusion auctions, our results could also be applied to general mechanisms with the following format: each agent $i$ is defined by two parameters $v_i$ and $z_i$. $v_i$ is $i$'s valuation function which is privately determined by a single parameter and $z_i$ is an environment-specific type which is independent of $v_i$ and $\emptyset\preceq z_i'\preceq z_i$ for any $z_i'\in \mathcal{Z}_i$, where $\mathcal{Z}_i$ is the misreport space of $z_i$ and $\preceq$ represents a partial order. For example, the well-known online mechanism design with no-early arrivals and no-late departures \cite{nisan2007algorithmic} essentially has similar underlying structures as diffusion auctions.

\newpage
\bibliographystyle{named}
\bibliography{tda}

\end{document}